\documentclass[10pt,twocolumn,twoside]{IEEEtran}

\pdfoutput=1
\usepackage{graphics} 
\usepackage{epstopdf}
\usepackage{epsfig} 
\usepackage{amsmath,amssymb,amsfonts,amsthm} 
\usepackage{mathtools,bbm}
\usepackage{tikz, epic,eepic}
\usetikzlibrary{shapes,arrows}
\usetikzlibrary{automata}

\usepackage{epsfig, amsbsy}
\usepackage{mathtools}
\usepackage{graphicx}
\usepackage{caption}
\usepackage{subcaption}
\theoremstyle{remark}

\newtheorem{theorem}{Theorem}
\newtheorem{proposition}{Proposition}
\newtheorem{corollary}{Corollary}
\newtheorem{lemma}{Lemma}
\newtheorem{definition}{Definition}
\newtheorem{example}{Example}
{\itshape}{\rmfamily}

\newtheorem{remark}{Remark}

\def\mcal{\mathcal}
\def\mbb{\mathbb}

\title{\LARGE \bf
Networked SIS Epidemics with Awareness
}

\author{Keith Paarporn$^{1}$, Ceyhun Eksin$^{1,2}$, Joshua S. Weitz$^{2,3}$ , Jeff S. Shamma$^{1,4}$ 
\thanks{$^{1}$School of Electrical and Computer Engineering, Georgia Institute of Technology, Atlanta, GA 30332
        {\tt\small kpaarporn@gatech.edu, ceyhuneksin@gatech.edu}}%
\thanks{$^{2}$School of Biology, Georgia Institute of Technology, Atlanta, GA}
\thanks{$^{3}$School of Physics, Georgia Institute of Technology, Atlanta, GA  {\tt\small jsweitz@gatech.edu}}%
\thanks{$^{4}$Computer, Electrical and Mathematical Sciences and Engineering, King Abdullah University of Science and Technology (KAUST) Thuwal, Kingdom of Saudi Arabia {\tt\small jeff.shamma@kaust.edu.sa}}%
}

\begin{document}

\maketitle
\thispagestyle{empty}
\pagestyle{empty}

\begin{abstract}

We study an SIS epidemic process over a static contact network where the nodes have partial information about the epidemic state. They react by limiting their interactions with their neighbors when they believe the epidemic is currently prevalent. A node's awareness is weighted by the fraction of infected neighbors in their social network, and a global broadcast of the fraction of infected nodes in the entire network. The dynamics of the benchmark (no awareness) and awareness models are described by discrete-time Markov chains, from which mean-field approximations (MFA) are derived. The states of the MFA are interpreted as the nodes' probabilities of being infected. We show a sufficient condition for existence of a ``metastable", or endemic, state of the awareness model coincides with that of the benchmark model. Furthermore, we use a coupling technique to give a full stochastic comparison analysis between the two chains, which serves as a probabilistic analogue to the MFA analysis. In particular, we show that adding awareness reduces the expectation of any epidemic metric on the space of sample paths, e.g. eradication time or total infections. We characterize the reduction in expectations in terms of the coupling distribution. In simulations, we evaluate the effect social distancing has on contact networks from different random graph families (geometric, Erd\H{o}s-Renyi, and scale-free random networks).
\end{abstract}
\section{Introduction}\label{sec:intro}

Mathematical models of epidemic spreading over networks have been extensively studied. Models characterize how the spatial features induced by network structure affects epidemic spread (\cite{Chak2003},\cite{Ganesh2005},\cite{scalefree_threshold},\cite{Newman},\cite{VM2009},\cite{Hassibi2013},\cite{Volz_2011}). The simplest formulation for such processes is the susceptible-infected-susceptible (SIS) model, where an individual is either infected or susceptible to infection. In such models, there is a threshold determining whether the epidemic eradicates quickly or persists for a long time. Specifically, $\delta/\beta > \lambda_{\text{max}}(A)$ ($\beta$ is the disease transmission rate, $\delta$ the healing rate, and $\lambda_{\text{max}}(A)$ the largest eigenvalue of the network adjacency matrix) is a sufficient condition for the disease to eradicate exponentially fast. The opposite strict inequality is a necessary condition for the disease to persist for a long period of time. The steady state in this regime is often referred to as the endemic or metastable state. How to devise control strategies in this regime is both an important research and policy question. 



A commonly studied control strategy is budgeted vaccine allocation, where the administration of vaccines among central nodes in the network optimally inhibits the epidemic (\cite{socialcost},\cite{Drakopoulos_2014},\cite{Preciado_2014}). In these situations, a central authority selects individuals to vaccinate and hence is required to have knowledge of the network structure. In game-theoretic settings, individuals decide for themselves whether or not to vaccinate based on an assessment of risks and benefits (\cite{VM_game},\cite{VM2015},\cite{mbah_2012},\cite{Molina2014},\cite{Bauch2004}). However, these models do not account for social behavior during the course of an epidemic, which can significantly slow epidemic spread without the aid of vaccines.

With the widespread availability of social media and news outlets on the internet and television, individuals may be well-informed about the current state of ongoing epidemics and how to take precautionary measures to avoid getting sick. In the recent 2009 H1N1 Influenza pandemic, people responded to public service announcements by increasing the frequency of washing hands, staying at home when they or loved ones were sick, or avoiding large public gatherings \cite{Steelfisher_2010}. In the recent Ebola outbreak in West Africa, a combination of quarantining and sanitary burial methods were shown to significantly reduce the rate of virus spread \cite{Pandey_2014}. These precautions and social distancing actions effectively limit epidemic spread. Individuals' distancing actions depend on the extent of how informed they are. The dissemination and exchange of information influences the public's behavior, affecting the course of the epidemic itself, in turn affecting the public's behavior again \cite{Bauch_2013}. This feedback loop allows epidemic spreading to coevolve with human social behavior, inducing complex dynamics. 


Recent research effort has focused on understanding the complexities that arise when incorporating human behavioral elements into existing models of epidemic spreading. A review of the recent literature can be found in \cite{Wang_2015_review}. Such models present general challenges for characterizing decentralized and dynamic protection measures and also capture a realistic aspect of disease spread in society. When individuals take social distancing actions based on the level of information they have, they reduce contact with others and the epidemic prevalence reduces significantly (\cite{Funk2009},\cite{Reluga_2010},\cite{Perra_2011}).  They can become aware of the epidemic by communicating with their social contacts or by a global broadcast (\cite{Arenas_2013},\cite{Arenas_2014}). Other actions include switching one's contact links, giving rise to a coevolving network \cite{Ogura_2015}.  Endowing individuals with local prevalence-based awareness highlights the role of network effects (\cite{Zhang_2014},\cite{Wu_2012}).

In this work, we study a networked SIS process with dynamically distributed information and social distancing actions. The information the agents receive comes from their social contacts and a global broadcast about the current state of the epidemic. An agent's social distancing action reduces its contact network interactions, the magnitude of which depends on how informed it is. We prove awareness reduces the endemic level, but cannot improve the epidemic threshold for persistence. In addition, we provide a stochastic comparison analysis between the awareness and benchmark (without awareness) processes using a coupling technique. This establishes an inequality between expectations of certain epidemic metrics (e.g. eradication time, cumulative infected), as well as the closed-form difference. We are also interested in studying which combinations of network structure and awareness are most effective. The results extend prior work by the same authors in \cite{Paarporn2015},\cite{Paarporn_2016}, where the mean-field approximation and coupling technique were initially studied.

The paper is organized as follows. In Section \ref{sec:models}, we describe a networked SIS epidemic process in discrete time, which we modify by incorporating a dynamic form of agent awareness and social distancing. In Section \ref{sec:MFA}, we introduce mean-field approximations on the probabilities of infection and prove the epidemic threshold for persistence with distancing remains the same as without. Section \ref{sec:coupling} provides a stochastic comparison analysis between the benchmark and awareness Markov chains through a coupling technique. In Section \ref{sec:sims}, we explore through simulations which random graph families are effective at restricting epidemic spread and prevalence when social distancing is a factor. Section \ref{sec:conclusion} gives concluding remarks. Proofs of some of the results are given in the Appendix.

\noindent{\bf Notation:} $\mbb{Z}_+ = \mbb{N}\cup\{0\}$ is the set of nonnegative integers. The all zeros and all ones vector in $\mbb{R}^n$ is written $\bold{0}_n$ and $\bold{1}_n$, respectively. For $x,y\in\mbb{R}^n$ we write $x\preceq y$  if $x_i \leq y_i$ for all $i$, and $x \prec y$ if $x_i < y_i$. To isolate a particular coordinate $i$, we write $x = (x_{-i},x_i) \in \mbb{R}^n$, where $x_{-i} = \{x_j : j \neq i\} \in \mbb{R}^{n-1}$. $\mcal{P}(E)$ is the power set of some set $E$. The complement of the set $E$ is written $E^c$. In probabilistic settings, we write $\chi_E(\cdot)$ as the indicator on the event $E$, i.e. $\chi_E(x) = 1$ if $x\in E$ and zero otherwise.

\section{Networked SIS Models}\label{sec:models}

\subsection{Benchmark SIS Model}

We introduce a model of epidemic spread which we refer to as the benchmark model (studied in \cite{Chak2003}, \cite{Hassibi2014}, and Section 5 of \cite{Hassibi2013}). Consider the set of nodes $\mcal{N}=\{1,\ldots, n\}$ interconnected by a set of edges $\mcal{E}_C$. Epidemic spread occurs in discrete time steps $t=0,1,\ldots$ over the undirected graph $\mcal{G}_C = (\mcal{N},\mcal{E})$, whose $n\times n$ adjacency matrix is defined for any $i,j\in\mcal{N}$, as  $[A_C]_{ij} = 1$ if $(i,j) \in \mcal{E}_C$ and $0$ otherwise. The graph $\mcal{G}_C$ is called the contact network. An agent $i\in\mcal{N}$ is either susceptible to the disease or infected by it. The epidemic states are defined as $\Omega \triangleq \{0,1\}^n$. For any $s\in\Omega$ and $i\in\mcal{N}$, either $s_i=0$, meaning agent $i$ is susceptible, or $s_i=1$, meaning it is infected. A susceptible node $i$ can contract the disease from neighboring agents in the contact network, $\mcal{N}_i^C \triangleq \{j\in\mcal{N} : (i,j) \in \mcal{E}_C\}$. When agent $i$ is susceptible in the epidemic state is $s\in\Omega$ $(s_i = 0)$, its probability of getting infected in the next time step due to an interaction with its neighbor $j\in\mcal{N}_i^C$ is given by $\beta s_j$ where $\beta \in (0,1)$ is the transmission probability of the disease. Hence, an individual can only contract the disease from an infected neighbor. Agent $i$ interacts with each of its neighbors independently. Therefore, $i$'s probability of not becoming infected in the next time step is
\begin{equation}
	p_{00}^i(s) \triangleq 1-p_{01}^i(s).
\end{equation}
Consequently, the probability $i$ becomes infected is 
\begin{equation}
	\label{eq:p01}
	p_{01}^i(s) \triangleq 1 - \prod_{j\in\mcal{N}_i^C} (1 - \beta s_j)
\end{equation}

 If $i$ is infected in state $s$ ($s_i=1$), it becomes susceptible in the next time step with probability $\delta p_{00}^i(s)$, where $\delta\in (0,1)$ is the healing probability. Thus, for an infected node to become susceptible, it must heal and not get re-infected by its neighbors. Agent $i$'s transition probabilities are summarized in Figure \ref{fig:nodemodeldiagram}, and described by $\mbb{P}_i : \Omega\times\{0,1\} \rightarrow [0,1]$ defined as
\begin{align}
	\text{If} \ s_i &= 0, \begin{cases} \mbb{P}_i(s,0) &= p_{00}^i(s) \\ \mbb{P}_i(s,1) &= p_{01}^i(s) \end{cases} \label{eq:bench_measure} \\
	\text{If} \ s_i &= 1, \begin{cases} \mbb{P}_i(s,0) &= \delta p_{00}^i(s) \\ \mbb{P}_i(s,1) &= 1-\delta p_{00}^i(s) \end{cases} \label{eq:bench_measure2}
\end{align}
For each $i\in\mcal{N}$ and $s\in\Omega$, the $\mbb{P}_i$ define the benchmark SIS Markov chain over $\Omega$ by the $2^n \times 2^n$ transition matrix $K$ with elements 
\begin{equation}
	\label{eq:bench_matrix}
	K(s,s') \triangleq \prod_{i=1}^n \mbb{P}_i(s,s_i'), \ \forall s,s' \in \Omega
\end{equation}
This chain has one absorbing state, the all-susceptible state $\bold{o} \triangleq \{0\}^n$. 
\begin{figure}[t]
		
		\begin{subfigure}[t]{\columnwidth}
			\centering
			\includegraphics{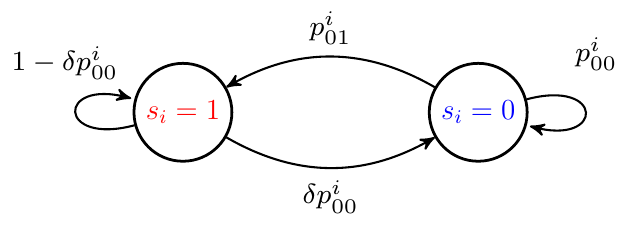}
	 		\caption{\small }
			\label{fig:nodemodeldiagram}
		\end{subfigure}
		
		\begin{subfigure}[t]{\columnwidth}
			\centering
			\includegraphics{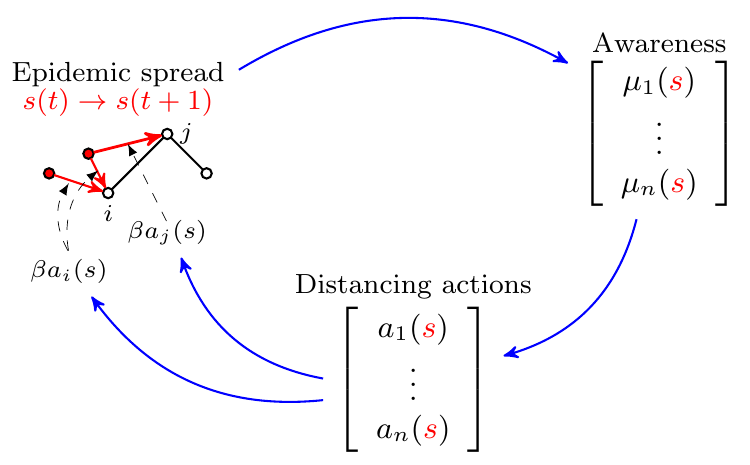}
	 		\caption{\small }
			\label{fig:systemdiagram}
		\end{subfigure}
		\caption{\small (a) Node-level state transition diagram. (b) System-level diagram }
\end{figure}
\subsection{Awareness SIS Model}

We modify the benchmark model to take into account the agents' awareness of the current epidemic state. The information agent $i$ receives comes from two sources: the proportion of infected neighbors in its local social network and a global broadcast of the proportion of infected nodes in the entire network. The social network is a graph $\mcal{G}_I = (\mcal{N},\mcal{E}_I)$ with the same nodes as $\mcal{G}_C$ but with different edges, representing the nodes' social communication links. The set of $i$'s neighbors in $\mcal{G}_I$ is written $\mcal{N}_i^I$. The information is given by
\begin{equation} \label{eq:mu}
	\mu_i(s) \triangleq \frac{\alpha}{|\mcal{N}_i^I|}\sum_{j\in\mcal{N}_i^I} s_j + \frac{1-\alpha}{n}\sum_{j=1}^n s_j \ ,\forall s \in \Omega
\end{equation}
where $\alpha\in[0,1]$ is a parameter that governs the trust nodes place in information from their social contacts. Consequently, node $i$ reduces its interactions with its physical  neighbors through the social distancing action
\begin{equation}
	\label{eq:a_i}
	a_i(s) \triangleq 1 - \mu_i(s),
\end{equation}
which reduces its susceptible-to-infected probability \eqref{eq:p01} to
\begin{equation}
	\label{eq:p01d}
	p_{01,d}^i(s) \triangleq 1 - \prod_{j\in\mcal{N}_i^C} (1 - \beta a_i(s) s_j).
\end{equation}
We similarly define $p_{00,d}^i(s) \triangleq 1 - p_{01,d}^i(s)$. An infected agent's probability of recovering becomes $\delta p_{00,d}^i(s)$. Note for all $s\in\Omega$, $p_{01,d}^i(s) \leq p_{01}^i(s)$. Combined with the social distancing behaviors $a_i$, the local awareness spread dynamics in \eqref{eq:mu} make the effect of user behavior on its infection probability endogenous to the benchmark chain model through a negative feedback loop. We define the $\mbb{P}_i^d$ analogously to \eqref{eq:bench_measure} and \eqref{eq:bench_measure2}:
\begin{align}
	\text{If} \ s_i &= 0, \begin{cases} \mbb{P}_i^d(s,0) &= p_{00,d}^i(s) \\ \mbb{P}_i^d(s,1) &= p_{01,d}^i(s) \end{cases} \label{eq:distancing_measure} \\
	\text{If} \ s_i &= 1, \begin{cases} \mbb{P}_i^d(s,0) &= \delta p_{00,d}^i(s) \\ \mbb{P}_i^d(s,1) &= 1-\delta p_{00,d}^i(s) \end{cases} \label{eq:distancing_measure2}
\end{align}
Thus, the $\mbb{P}_i^d$ define the distancing Markov chain over $\Omega$ by the transition matrix $K_d$ with elements 
\begin{equation}
	\label{eq:distancing_matrix}
	K_d(s,s') \triangleq \prod_{i=1}^n \mbb{P}_i^d(s,s_i'), \ \forall s,s' \in \Omega
\end{equation}
whose unique absorbing state is also  $\bold{o}$, the all-susceptible state. The feedback between social distancing actions and epidemic states is illustrated in Figure \ref{fig:systemdiagram}.
\begin{remark}
Our model of awareness captures the different ways an agent may receive information about an ongoing epidemic from the media. Large media corporations and public health institutions such as the Centers for Disease Control and Prevention (CDC) and the World Health Organization (WHO) often report an estimated total number of people infected nationwide or globally at a given time, and this information is disseminated amongst the population. Information is also exchanged through one's personalized social links, which can range beyond a person's geographic location. Thus, a node's awareness is composed of a linear combination of both sources of information, as given in \eqref{eq:mu}.
\end{remark}

\section{Mean-field Approximations}\label{sec:MFA}

\subsection{Derivation}
We derive a mean-field approximation (MFA) of the Markovian dynamics described in the previous section. The MFA is a deterministic, discrete-time dynamic system with an $n$-dimensional state space  $[0,1]^n$, which is interpreted to be each node's probability of being infected at any given time. 

Here, we write $s^t \in \Omega$ as the epidemic state at time $t = 0,1,\ldots$. Indeed, consider the node-level stochastic state transition update
\begin{equation}
	\label{eq:stoch_update}
	s_i^{t+1} = s_i^{t}B(1-\delta p_{00,d}^i(s^t))+ (1-s_i^{t})B(p_{01,d}^i(s^t))
\end{equation}
of the distancing chain where $B(\lambda)$ denotes a Bernoulli random variable with parameter $\lambda\in [0,1]$. For shorthand, we write $x_i^t \triangleq \text{Pr}(s_i^t = 1)$ for the probability node $i$ is infected at time $t$. Taking the probability of both sides equaling one, 
\begin{align}
	\label{eq:LTP}
	x_i^{t+1} &= \text{Pr}(B(1-\delta p_{00,d}^i(s^t)) = 1 | s_i^t = 1) \text{Pr}(s_i^t=1) + \cdots \nonumber \\
	& \ \ \ \ \text{Pr}(B(p_{01,d}^i(s^t))=1 | s_i^t=0)\text{Pr}(s_i^t=0) \nonumber \\
	&= x_i^t(1-\delta p_{00,d}^i(s^t)) + (1-x_i^t)p_{01,d}^i(s^t) \nonumber \\
	&= x_i^t(1-\delta) + (1-(1-\delta)x_i^t)p_{01,d}^i(s^t) 
\end{align}
Note the expression for $x_i^{t+1}$ still depends stochastically on the state $s^t$.  To obtain a mean-field approximation of $x_i^{t+1}$, we simply replace the state $s^t$ with $x^t$ (the $n$-vector with components $x_i^t$) in \eqref{eq:LTP}. Thus, by redefining $x_i^{t+1}$ to obey this approximation and extending the domain of $\mu_i(\cdot), a_i(\cdot)$ and $p_{01,d}^i(\cdot)$ from $\{0,1\}^n$ to $[0,1]^n$, we have the following approximate dynamics 
\begin{equation}
	\label{eq:MFA_dist}
	x_i^{t+1} = x_i^t(1-\delta) + (1-(1-\delta)x_i^t) p_{01,d}^i(x^t)
\end{equation}
on the time evolution of node $i$'s probability of being infected. Stacking the dynamics for each node into a vector, we obtain a mapping $\phi : [0,1]^n \rightarrow [0,1]^n$ where $\phi_i(x) = x_i(1-\delta) + (1-(1-\delta)x_i)p_{01,d}^i(x)$ and $x^{t+1} = \phi(x^t)$. This MFA of the distancing chain is in contrast to the MFA of the benchmark chain, 
\begin{equation}
	\label{eq:MFA_standard}
	x_i^{t+1} = x_i^t(1-\delta) + (1-(1-\delta)x_i^t) p_{01}^i(x^t)
\end{equation}
which is studied thoroughly in \cite{Hassibi2013}. It is derived in the same manner, and is described by the mapping $\psi : [0,1]^n \rightarrow [0,1]^n$ with $\psi_i(x) = x_i(1-\delta) + (1-(1-\delta)x_i)p_{01}^i(x)$ and $x^{t+1} = \psi(x^t)$.

We make a few remarks on the basic structure of these two mappings. They are nonlinear, continuous mappings satisfying $\phi(x) \prec \psi(x)$ whenever $x\in [0,1]^n\backslash\bold{0}_n$ and $\alpha\in [0,1)$. Also, $\phi(\bold{0}_n) = \psi(\bold{0}_n) = \bold{0}_n$. Linearization of $\phi$ and $\psi$ about the origin yields the same Jacobian matrix, $\beta A_C + (1-\delta)I$, and hence the same linearized dynamics $x^{t+1} = (\beta A_C + (1-\delta)I)x^t$. The linear dynamics serve as an upper bound to both \eqref{eq:MFA_dist} and \eqref{eq:MFA_standard}. Therefore, if $\lambda_{\text{max}}(\beta A_C + (1-\delta)I) < 1$, the origin is a globally stable fixed point and it is an unstable fixed point if $\lambda_{\text{max}}(\beta A_C + (1-\delta)I) > 1$.

\subsection{Existence of a Non-trivial Fixed Point}
We now provide a sufficient condition for the existence of a non-trivial fixed point ($\succ \bold{0}_n$, the $n$-vector of zeros) of $\phi$.  
\begin{theorem}\label{theorem_nontrivial_fixed_point}  If $\lambda_{\text{max}}(\beta A_C + (1-\delta)I_n) > 1$, there exists a nontrivial fixed point for $\phi$. 
\end{theorem}
The existence of such a fixed point suggests the epidemic has an endemic state, where the disease spreads fast enough to sustain an epidemic in the network. Our condition coincides with the condition for existence, uniqueness, and global asymptotic stability of the non-trivial fixed point $q^*$ of $\psi$, which is $\lambda_{\text{max}}(\beta A_C + (1-\delta)I) > 1$, i.e. when the origin in the linearized dynamics is unstable (Theorem 5.1, \cite{Hassibi2013}). This condition incorporates the factors that contribute to the rate of spreading -  $\delta,\beta$, and the contact network $A_C$. 
The proof makes use of the following lemmas.
\begin{lemma}[Lemma 3.1, \cite{Hassibi2013}] \label{lemma_threshold}
	There exists a vector $\nu \succ 0_n$ such that $(\beta A_C - \delta I_n)\nu \succ 0_n$ if and only if $\lambda_{\text{max}}(\beta A_C + (1-\delta)I_n) > 1$.
\end{lemma}
The connectedness assumption for $\mcal{G}_C$ is necessary for the above Lemma because the proof applies the Perron-Frobenius theorem for nonnegative irreducible matrices.
The next result is an equivalent formulation of Brouwer's fixed point theorem. 
\begin{lemma}\label{Brouwer}(Theorem 4.2.3, \cite{FPT})\label{lemma_fixed_point_existence}: Suppose $f_i : D_n \rightarrow \mathbb{R}, i = 1,\ldots,n$ are continuous mappings, where 
\[ D_n = \{x \in \mathbb{R}^n : x_i \in [\ell_i,u_i], \ \forall i \} \]
for real numbers $\ell_i,u_i$. We also define the set
\[ D_{-i} = \{ x_{-i} \in \mathbb{R}^{n-1} : x_j \in [\ell_j,u_j] \ \forall j \neq i \} \] 
If for every $i$ and for all $x_{-i} \in D_{-i}$,
\begin{align}
	f_i(x_1,\ldots,\ell_i,\ldots,x_n)  &= f_i(x_{-i},\ell_i) \geq 0 \label{eq:loweredge} \\
	f_i(x_1,\ldots,u_i,\ldots,x_n) &= f_i(x_{-i},u_i) \leq 0 \label{eq:upperedge} ,
\end{align}
then there exists a point $x^* \in D_n$ such that $f_i(x^*) = 0$, for all $i = 1,\ldots,n$. 
\end{lemma}
The final lemma needed is a technical result for the mean-field mappings $\phi_i$.
\begin{lemma}\label{cstar}
	For each $i\in\mcal{N}$, define the maps $f_i : [0,1]^n \rightarrow \mbb{R}$
	\begin{align} 
		f_i(x) &\triangleq \phi_i(x) - x_i \nonumber \\  
		&= -\delta x_i + (1-(1-\delta)x_i)p_{01,d}^i(x)
	\end{align}
	Then for any $i\in\mcal{N}$ and $x_{-i} \in [0,1]^{n-1}$, the function $f_i(x_{-i},\cdot)$ has a unique root $c_i^*(x_{-i}) \in [0,1)$ which depends continuously on $x_{-i}$. Furthermore, one can find a sequence $x_{-i}^k \rightarrow \bold{0}_{n-1}$ s.t. $c_i^*(x_{-i}^k)$ is monotonically decreasing to 0.
\end{lemma}
\begin{proof}
	For any $i\in\mcal{N}$ and $x_{-i} \in [0,1]^{n-1}$,
	\begin{equation}\label{eq:edge}
		f_i(x_{-i},0) = p_{01,d}^i(x_{-i},0) \geq 0.
	\end{equation}
	and 
	\begin{equation}\label{eq:negative_corner}
		f_i(x_{-i},1) = \delta(p_{01,d}^i(x_{-i},1) - 1) < 0.
	\end{equation}
	The function $f_i(x_{-i},\cdot)$ is strictly decreasing: for $a,b \in [0,1]$ s.t. $a < b$, $f_i(x_{-i},a) - f_i(x_{-i},b)$ is given by
	\begin{align}
		(b-a)&(\delta + (1-\delta)p_{01,d}^i(x_{-i},b)) + \cdots \nonumber \\
		&+ (1-a(1-\delta))( p_{01,d}^i(x_{-i},a) - p_{01,d}^i(x_{-i},b)) \nonumber \\
		&> 0. \nonumber
	\end{align}
	This follows because $p_{01,d}^i(x_{-i},x_i)$ is decreasing in $x_i$ ($x_i$ contributes to global awareness). Hence for every $x_{-i} \in [0,1]^{n-1}$, there is a unique $c_i^*(x_{-i}) \in [0,1)$ s.t. $f_i(x_{-i},c_i^*(x_{-i})) = 0$, and $c_i^*(x_{-i})$ depends continuously on $x_{-i}$. To see this, observe that $c_i^*(x_{-i}) \in [0,1)$ is a root of the equation
	\begin{equation}
		(1-(1-\delta)x_i)\left(1 - \prod_{j\in\mcal{N}_i^C}(1-a_i(x_{-i},x_i)\beta x_j) \right) - \delta x_i = 0,
	\end{equation}
	which is a polynomial in $x_i$. The coefficients of the polynomial depend continuously on $x_{-i} \in [0,1]^{n-1}$, and the roots of any polynomial are continuous with respect to its coefficients. Consequently, for any sequence $x_{-i}^k \rightarrow \bold{0}_{n-1}$, $c_i^*(x_{-i}^k) \rightarrow c_i^*(\bold{0}_{n-1}) = 0$ by continuity. This allows us to select a subsequence of $x_{-i}^k$ such that $c_i^*$ is monotonically decreasing along the subsequence. 
\end{proof}
\noindent We are now ready to prove the main result of this section.
\begin{proof}[Proof of Theorem \ref{theorem_nontrivial_fixed_point}:]
Let the mappings $f_i$, $i\in\mcal{N}$ be as in Lemma \ref{cstar}. We need to verify \eqref{eq:loweredge} and \eqref{eq:upperedge} hold for all $i$ and for choices of $\ell_i,u_i$ satisfying $0 < \ell_i < u_i$. This ensures the awareness dynamic $\phi$ has a fixed point other than the origin. 
\begin{figure}[t]
	\centering
	\includegraphics{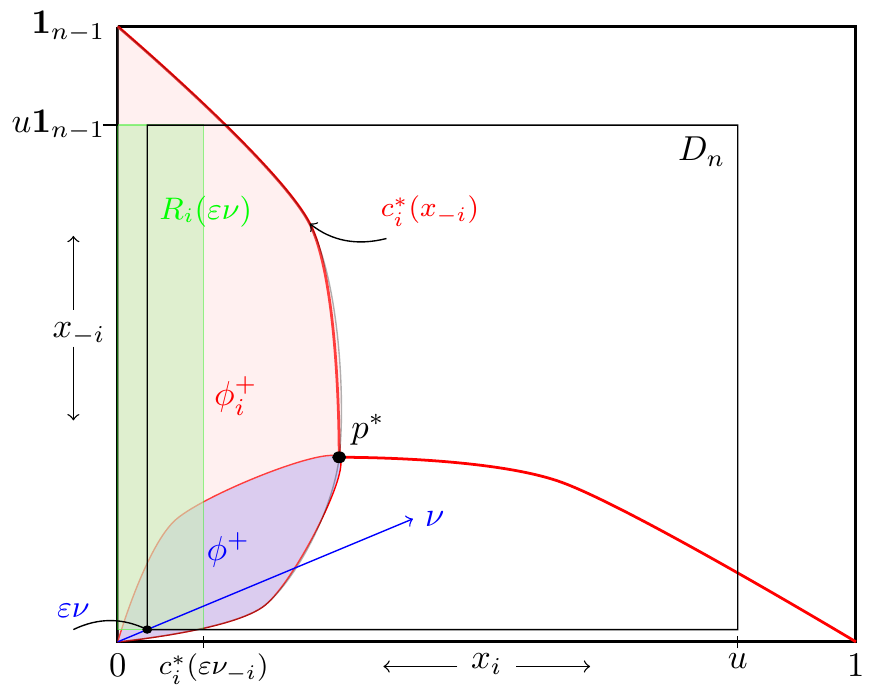}
	\caption{Diagram of the proof of Theorem \ref{theorem_nontrivial_fixed_point}. Here, $p^*$ denotes a nontrivial fixed point of $\phi$.}
	\label{fig:proof}
\end{figure}
Choose $u_i = u$ where $u$ satisfies
\begin{equation}
	\max_{i\in\mcal{N}} \max_{x_{-i}\in [0,1]^{n-1}}c_i^*(x_{-i}) < u < 1.
\end{equation}
Then for all $x_{-i} \in [0,u]^{n-1}$,
\begin{equation}
	f_i(x_{-i},u) < 0
\end{equation}
since $u > c_i^*(x_{-i})$. Thus, \eqref{eq:upperedge} holds, regardless of the choice of $\ell_i$. However, it remains to find the $\ell_i > 0$ s.t. \eqref{eq:loweredge} is satisfied. Let $f(x) \triangleq [f_1(x),\ldots, f_n(x)]^T$ and define the sets
\begin{equation} \label{eq:phiplus}
		\phi^+_i \triangleq \{x \in [0,1]^n : f_i(x) \geq 0 \}, \ \ \phi^+ \triangleq \bigcap_{j=1}^n \phi_j^+ 
\end{equation}
The Jacobian of $f$ about the origin is $(\beta A_C - \delta I_n)$. By Lemma \ref{lemma_threshold}, there exists a vector $\nu \succ \bold{0}_n$ such that $(\beta A_C - \delta I_n) \nu \succ \bold{0}_n$. Consequently for sufficiently small $\varepsilon > 0$,
\begin{equation}\label{eq:derivative}
	f(\varepsilon\nu) \succ \bold{0}_n,
\end{equation}
or $\varepsilon\nu \in \phi^+$. We also define the set
\begin{equation}
	R_i(x_{-i}) \triangleq \{ y \in \mbb{R}^n : y_i \in [0,c_i^*(x_{-i})], y_j \in [x_j,u], j\neq i\}
\end{equation}
If  $\varepsilon\nu_i \leq c_i^*(\varepsilon\nu_{-i})$ and $R_i(\varepsilon \nu) \subset \phi_i^+$, then \eqref{eq:loweredge} is satisfied, i.e. $f_i(x_{-i},\varepsilon\nu_i) \geq 0$ on $D_{-i} = \{\varepsilon\nu_{-i} \preceq x_{-i} \preceq u\bold{1}_{n-1}\}$ for $\varepsilon$ sufficiently small. We already have $\varepsilon\nu_i \leq c_i^*(\varepsilon\nu_{-i})$ because $\varepsilon\nu \in \phi^+ \subset \phi_i^+$. To show $R_i(\varepsilon\nu) \subset \phi_i^+$, by Lemma \ref{cstar} we can find a sequence $\varepsilon_k\nu \in \phi^+$ with $\varepsilon_k \rightarrow 0$ s.t. $c_i^*(\varepsilon_k\nu_{-i})$ is monotonically decreasing to 0. By stopping at a large enough $k$, we can take a $\varepsilon$ small enough such that 
\begin{equation}\label{eq:smallest}
	c_i^*(\varepsilon \nu_{-i}) = \min_{x_{-i}\in D_{-i}} c_i^*(x_{-i}) 
\end{equation}
Consequently, $R_i(\varepsilon\nu) \subset \phi_i^+$. Choosing $\varepsilon$ small enough to satisfy \eqref{eq:smallest} for all $i\in\mcal{N}$ verifies \eqref{eq:loweredge} by using $D_n = \{x \in \mbb{R}^n : x_j \in [\varepsilon\nu_j,u]\}$. See Figure \ref{fig:proof} for an illustration of the proof. By Lemma 2, $\phi$ has a fixed point contained in $D_n$.
\end{proof}
The condition of Theorem \ref{theorem_nontrivial_fixed_point} is independent of the awareness parameter $\alpha$ and the structure of the information network $\mathcal{G}_I$. Hence, social distancing alone cannot restore stability of the disease-free equilibrium point. However, social distancing lowers the overall metastable state of an epidemic.
\begin{corollary}\label{MFA_corollary}
If $q^* \succ \bold{0}_n$ is the unique nontrivial fixed point of $\psi$, then any nontrivial fixed point $p^*$ of $\phi$ satisfies $p^* \prec q^*$ whenever $\alpha \in [0,1)$. 
\end{corollary}
\begin{proof}
Define the sets $\psi_i^+$ and $\psi^+$ similarly as in \eqref{eq:phiplus}. It was shown in \cite{Hassibi2013} that $q^*$ is the unique maximal element of $\psi^+$, i.e $q^* \succ q$, $\forall q \in \psi^+$, $q \neq q^*$. Observe $\phi(x) \prec \psi(x)$ for any $x \in [0,1]^n \backslash \bold{0}_n$.  Let $x \in \phi^+$, $x \neq \bold{0}_n$. Then $x \preceq \phi(x) \prec \psi(x)$, so $x \in \psi^+$. Therefore, $\phi^+ \subset \psi^+$, and since $p^* \in \phi^+$ for any nontrivial fixed point of $\phi$, $p^* \prec q^*$.
\end{proof}
The mean-field analysis of this section reveals the qualitative dynamics not addressed by an absorbing Markov chain analysis. Since the all-susceptible state is the unique absorbing state and accessible from every other state, the disease eradicates in finite time with probability one. This answers what happens in the long-run, whereas the MFA analysis answers what happens before this eventuality. The MFA analysis concurs with what is observed in simulations of the Markov chain dynamics - fast convergence to an endemic ``metastable" state that persists for a very long time before eradication. Numerical simulations suggest stability of one particular nontrivial fixed point $p^*$ of $\phi$.  A characterization of these fixed points for different random graph families is given in Section \ref{sec:sims}. 

\section{Stochastic Comparison Between Awareness and Benchmark Model}\label{sec:coupling}

In this section, we use the properties of monotone couplings to prove that adding awareness reduces the expectation of any increasing random variable quantifying an epidemic metric (e.g. absorption time, total infections), and that the benchmark chain dominates the distancing chain in terms of a partial ordering on sample paths. These results serve as a probabilistic analogue to the conclusions made in Section \ref{sec:MFA} about the strict ordering of the nontrivial fixed points in the corresponding mean-field approximations. Here, we have a closed form expression for the reduction whereas in Corollary \ref{MFA_corollary}, only an inequality relation is established. First, we provide relevant definitions and basic preliminary results of monotone couplings. For a full reference on monotone coupling, see Ch 4 of \cite{lindvall2002lectures}.


\subsection{Monotone couplings}
Consider a general countable space $X$. Recall a partially ordered set $(X,\preceq_X)$ is the set $X$ together with a relation $\preceq_X$ among its elements which satisfies for all $x,y$, and $z \in X$,
\begin{itemize}
	\item $x \preceq_X x$
	\item If $x\preceq_X y$ and $y \preceq_X z$, then $x\preceq_X z$
	\item If $x \preceq_X y$ and $y \preceq_X x$, then $x=y$.
\end{itemize}
\begin{definition}\label{coupling_def}
	Let $p_1,p_2$ be probability measures on a measurable space $(X,\mcal{F})$ and suppose $(X,\preceq_X)$ is a partially ordered set. A \emph{monotone coupling} of $p_1,p_2$ is a probability measure $p$ on $(X^2,\mcal{F}^2)$ such that for all $x',y' \in X$,
	\begin{equation}
		\label{eq:marg_probs}
		\sum_{x \preceq_X y'} p(x,y') = p_2(y') \ \text{and} \ \sum_{y \succeq_X x'} p(x',y) = p_1(x').
	\end{equation}
\end{definition}
Thus, for any $x,y \in X$ s.t $x \not\preceq_X y$, $p(x,y) = 0$, and the marginals of $p$ are $p_1,p_2$. 
\begin{example}\label{biased_coins}
	For an illustrative example of monotone coupling, consider two biased coins where the biases $q_A,q_B < 1$ for landing heads satisfy $q_A < q_B$. The coupling is a joint distribution assigned to the pair of coin flips to ensure the $q_A$ coin can never land heads with the $q_B$ coin landing tails, while the marginal coin flip probabilities remain the same. Specifically, define $p_X(0) = 1-q_X, p_X(1) = q_X$ for $X \in \{A,B\}$. Also, define $p_{AB} : \{0,1\}^2 \rightarrow [0,1]$ by
	\begin{equation}\label{eq:coins}
		\begin{cases}
			p_{AB}(0,0) &= \ 1-q_B \\
			p_{AB}(0,1) &= \ q_B-q_A \\
			p_{AB}(1,0) &= \ 0 \\
			p_{AB}(1,1) &= \ q_A
		\end{cases}
	\end{equation}
	Checking \eqref{eq:marg_probs}, the marginals are such that $\sum_{b \geq 1} p_{AB}(1,b) = q_A$, $\sum_{a \leq 1} p_{AB}(a,1) = q_B$ and  $\sum_{b \geq 0} p_{AB}(0,b) = 1-q_A$, $\sum_{a \leq 0} p_{AB}(a,0) = 1- q_B$ . Thus, $p_{AB}$ is a monotone coupling of $p_A,p_B$. 
\end{example}

We say a function $Z : X \rightarrow \mbb{R}$ is \emph{increasing in} $X$ if whenever $x \preceq_X y$, $Z(x) \leq Z(y)$. The next result characterizes the difference in expectations of increasing random variables between the marginals of a monotone coupling.
\begin{proposition}\label{incr_prop}
	Keeping the notation of Definition \ref{coupling_def}, suppose $p$ is a monotone coupling of $p_1,p_2$. If a random variable $Z : X \rightarrow \mbb{Z}_+$ is increasing in $X$, then 
\begin{equation}\label{eq:difference}
	\mbb{E}_{p_2}(Z) - \mbb{E}_{p_1}(Z) = \sum_{\tau=0}^\infty p(Z_\tau^c,Z_\tau)
\end{equation}
where $Z_\tau = \{x : Z(x) > \tau\}$.
\end{proposition}
\begin{proof} 
	Consider the following quantities:
	\begin{align}
		p(Z_\tau,Z_\tau) &= \sum_{x\in Z_\tau}\sum_{y\in Z_\tau} p(x,y) \nonumber \\
		&= \sum_{x\in Z_\tau}\sum_{y \succeq_X x} p(x,y) = p_1(Z_\tau) \label{eq:incr_set} \\
		p(X,Z_\tau) &= p_2(Z_\tau)
	\end{align}
	The second sum over $\{y\in Z_\tau\}$ can be replaced with $\{y\succeq_X x\}$ in \eqref{eq:incr_set} because 1) for any $x\in Z_\tau$, we have $\{y : y\succeq_X x\} \subset Z_\tau$; and 2) since $p$ is a monotone coupling, for any $y\in Z_\tau$ s.t. $y \not\succeq_X x$, $p(x,y) = 0$. The last equality of \eqref{eq:incr_set} follows from \eqref{eq:marg_probs}. Since $(X,Z_\tau) \supset (Z_\tau,Z_\tau)$ we can write
	\begin{align}
		p_2(Z_\tau) - p_1(Z_\tau) &= p(X,Z_\tau) - p(Z_\tau,Z_\tau) \nonumber \\
		&= p((X,Z_\tau)\backslash (Z_\tau,Z_\tau)) \nonumber \\
		&= p(Z_\tau^c,Z_\tau) \nonumber 
	\end{align}
	Equation \eqref{eq:difference} immediately follows.
\end{proof}
\begin{example}
	Consider the biased coins of Example \ref{biased_coins}. One can extend this example to sequences of $m\geq 2$ flips, $\{0,1\}^m$ with the partial order $x \preceq y$ if $x_i \leq y_i$, $i = 1,\ldots,m$, for $x,y \in \{0,1\}^m$. Define
\begin{align}
	\overline{p}_X(x) &= \prod_{k=1}^m p_X(x_k) , \ X \in \{A,B\} \\
	\overline{p}_{AB}(x,y) &= \prod_{k=1}^m p_{AB}(x_k,y_k).
\end{align}
Then $\overline{p}_{AB}$ is a monotone coupling of $\overline{p}_A,\overline{p}_B$. For $x\in \{0,1\}^m$, let $Z(x) = \sum_{i=1}^m x_i$ be the random variable of the number of heads for any given toss sequence. Then $Z$ is increasing in $\{0,1\}^m$. By Proposition \ref{incr_prop}, 
\begin{equation}
	\mbb{E}_{\overline{p}_B}(Z) - \mbb{E}_{\overline{p}_A}(Z) = \sum_{\tau=0}^m \overline{p}_{AB}(Z_\tau^c,Z_\tau).
\end{equation}
Of course, one could trivially compute the above as $m(q_B - q_A)$ since the distribution of $Z$ is Bernoulli. However, Proposition \ref{incr_prop} generalizes the difference for any increasing $\mbb{Z}_+$-valued random variable over a partially ordered set.
\end{example}

The notion of stochastic domination is also relevant in our comparison analysis.
\begin{definition}
	An \emph{upper set} $\mcal{I}$ is a non-empty subset of $(X,\preceq_X)$ that satisfies the following property: if $x\in\mcal{I}$ and $y\succeq_X x$, then $y\in\mcal{I}$. Let $p_1,p_2$ be two probability measures on $(X,\mcal{F})$. Then $p_2$ \emph{stochastically dominates} $p_1$, written as $p_2 \succeq p_1$, if for any upper set $\mcal{I} \subset X$, $p_1(\mcal{I}) \leq p_2(\mcal{I})$. 
\end{definition}
Our comparison between benchmark and distancing chains falls into the framework of the above analysis.
\subsection{Sample path comparison analysis}
Our main result provides a construction of a monotone coupling between the benchmark and distancing probability distributions on sample paths. 

\begin{definition}
	A \emph{sample path} is a sequence $g = \{g^t\}_{t\in\mbb{Z}_+}$ such that $g^t \in\Omega$ and $K(g^t,g^{t+1}) > 0$ for all $t\geq 0$, and there is a $T < \infty$ such that $g^T = \bold{o}$. The set of sample paths is denoted by $\Gamma$.
\end{definition}
The \emph{absorption time} $T : \Gamma \rightarrow \mbb{Z}_+$  of a sample path $g$ is given by 
\begin{equation}\label{eq:T}
	T(g) \triangleq \min\{t : g^t = \bold{o}\}.
\end{equation}
Thus for all $g\in\Gamma$, $T(g) < \infty$. Also, $g^t = \bold{o}$ and $K(g^t,g^{t+1}) = 1$ for all $t \geq T(g)$. Note that $\Gamma$ is countable since it is the countable union of the finite sets $\{g : T(g) = t\}$ for $t = 0,1,2,\ldots$.

                                                                                                                                                                                                                                                                                                                                                                                                   The distribution $\mu_\pi : \mcal{P}(\Gamma) \rightarrow [0,1]$ on sample paths under the benchmark SIS chain with starting distribution $\pi\in\Delta(\Omega)$ is given, for any $A \in \mcal{P}(\Gamma)$, by 
\begin{equation}
	\label{eq:bench_measure_gamma}
	\mu_\pi(A) \triangleq \sum_{g\in A}\pi(g^0)\prod_{t=0}^{T(g)-1} K(g^t,g^{t+1})
\end{equation}
and similarly under the distancing chain by
\begin{equation}
	\label{eq:distancing_measure_gamma}
	\nu_\pi(A) \triangleq \sum_{g\in A} \pi(g^0)\prod_{t=0}^{T(g)-1} K_d(g^t,g^{t+1}).
\end{equation}
Also, note that $\Gamma$ is defined to exclude the set of sample paths that are never absorbed, $\{g:g^t \neq \bold{o}, \forall t\in\mbb{Z}_+\}$. These are infinite sequences that never terminate, and therefore are uncountable. The probabilities $\mu_\pi,\nu_\pi$, however are well-defined on $\Gamma$ without such sample paths: 
\begin{align}
	\sum_{g\in\Gamma} \mu_\pi(g) &= \sum_{t=0}^\infty \mu_\pi(\{g : T(g) = t\}) \\
	&= \sum_{s\in\Omega} \pi(s) \sum_{t=0}^\infty r_s(Q^t) - r_s(Q^{t+1}) \\
	&= 1
\end{align}
Here, $Q$ is the $2^{n}-1 \times 2^{n}-1$ sub-stochastic matrix of transition probabilities between non-absorbing states, and $r_s(Q)$ is the $s^{\text{th}}$ row-sum of $Q$. Hence, $r_s(Q^t) - r_s(Q^{t+1})$ is the probability a sample path starting from state $s$ is absorbed at time $t$. The elements of $Q^t$ approach zero as $t \rightarrow \infty$.
\begin{remark}
	($\Omega,\preceq_\Omega$) is a partially ordered set. For $s,s'\in\Omega$, $s \preceq_\Omega s'$ if $s_i \leq s'_i$ for all $i\in\mcal{N}$.
\end{remark}
\begin{remark}
	($\Gamma,\preceq_\Gamma$) is a partially ordered set. For $h,g \in\Gamma$, $h \preceq_\Gamma g$ if $h^t \preceq_\Omega g^t$ for all $t\in\mbb{Z}_+$.
\end{remark}
\begin{figure}[t]
	\centering
		\centering
		\includegraphics{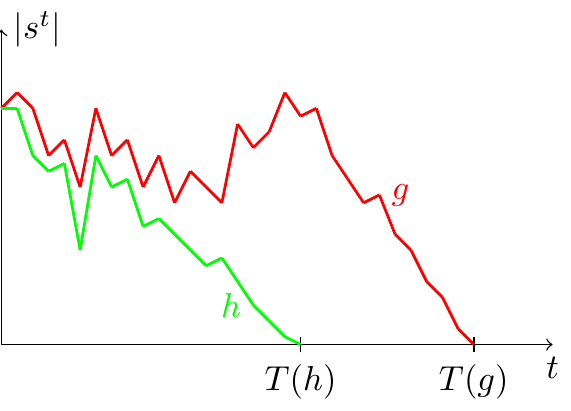}
	 	\caption{\small A pair of sample paths $(h,g)$ drawn from $\Phi_\pi$. }
		\label{fig:hgcoupling}
\end{figure}

Next, we present the main result of this section, which constructs a monotone coupling distribution of $\nu_\pi,\mu_\pi$ by exploiting the differences in node-level transition probabilities. 

\begin{theorem}\label{Phi_coupling}
	Suppose $x,y\in\Omega$ with $x \preceq_\Omega y$. For each $i\in\mcal{N}$, define $\varphi_i^{x,y} : \{0,1\}^2 \rightarrow [0,1]$ according to
	\begin{align}
		x_i &= y_i = 1, \ \ \
		\begin{cases}\label{eq_infected_infected}
			\varphi_i^{x,y}(0,0) &= \delta(1 - p_{01}^i(y)) \\ 
			\varphi_i^{x,y}(0,1) &= \delta(p_{01}^i(y) - p_{01,d}^i(x)) \\ 
			\varphi_i^{x,y}(1,0) &= 0 \\ 
			\varphi_i^{x,y}(1,1) &= 1 - \delta(1 - p_{01,d}^i(x)) 
		\end{cases}  \\
		x_i &= y_i = 0, \ \ \
		\begin{cases} 
			\varphi_i^{x,y}(0,0) &= 1-p_{01}^i(y) \\ 
			\varphi_i^{x,y}(0,1) &= p_{01}^i(y) - p_{01,d}^i(x) \\ 
			\varphi_i^{x,y}(1,0) &= 0 \\ 
			\varphi_i^{x,y}(1,1) &= p_{01,d}^i(x) 
		\end{cases}  \\
		x_i&=0,y_i=1, 
		\begin{cases} \label{eq:delta_mod}
			\varphi_i^{x,y}(0,0) &= \delta(1 - p_{01}^i(y)) \\ 
			\varphi_i^{x,y}(0,1) &= 1 - p_{01,d}^i(x) - \delta(1 - p_{01}^i(y)) \\ 
			\varphi_i^{x,y}(1,0) &= 0 \\
			\varphi_i^{x,y}(1,1) &= p_{01,d}^i(x)
		\end{cases}
	\end{align}
	Also, define $\varphi^{x,y} : \Omega^2 \rightarrow [0,1]$ for $x\preceq_\Omega y$ by
	\begin{equation}\label{eq:phi_omega}
		\varphi^{x,y}(\omega,\omega') \triangleq \prod_{i=1}^n \varphi^{x,y}_i(\omega_i,\omega_i') \ \ \forall \omega,\omega'\in\Omega.
	\end{equation}
	Lastly, define $\Phi_\pi : \Gamma^2 \rightarrow [0,1]$ for any $\pi\in\Delta(\Omega)$ by
	\begin{equation}\label{eq:phi_pi}
		\Phi_\pi(h,g) \triangleq \chi(h^0=g^0)\pi(h^0)\prod_{t=0}^{T(g)-1} \varphi^{h^t,g^t}(h^{t+1},g^{t+1}).
	\end{equation}
	Then $\Phi_\pi$ is a monotone coupling of $\nu_\pi,\mu_\pi$.
\end{theorem}
\begin{proof}
See Appendix.
\end{proof}
The couplings between node-level transition probabilities $\varphi_i^{x,y}$ given in \eqref{eq_infected_infected}-\eqref{eq:delta_mod} are used to establish a coupling between benchmark and distancing probability distributions on sample paths in \eqref{eq:phi_pi}. The form of the $\varphi_i^{x,y}$ is identical to the method in which two biased coins are coupled in \eqref{eq:coins}.
When the coupling rule is applied to each node's probability of infection, it ensures no node can be infected in the distancing chain while being susceptible in the benchmark chain. Consequently, the monotone coupling $\Phi_\pi$ is a distribution on pairs of sample paths $(h,g)$ satisfying $h^0 = g^0$, $h\preceq_\Gamma g$ (see Figure \ref{fig:hgcoupling}) and marginally, $h\sim\nu_\pi$, $g\sim\mu_\pi$. The next result characterizes the difference between $\mu_\pi$ and $\nu_\pi$-expectations of any non-negative increasing function on the sample paths $\Gamma$ with respect to the coupling distribution $\Phi_\pi$. 
\begin{corollary}
	For any increasing $\mbb{Z}_+$-valued random variable $Z$ in $\Gamma$,
	\begin{equation}
		\label{eq:thm_gap}
		\mbb{E}_{\mu_\pi}(Z) - \mbb{E}_{\nu_\pi}(Z) = \sum_{\tau=0}^\infty \Phi_\pi(Z_\tau^c,Z_\tau).
	\end{equation}
	where $Z_\tau = \{x \in \Gamma : Z(x) > \tau\}$.
\end{corollary}
\begin{proof}
	Immediate from Theorem \ref{Phi_coupling} and Proposition \ref{incr_prop}.
\end{proof}
One can think of an increasing $Z:\Gamma \rightarrow \mbb{Z}_+$ as an epidemic cost metric. Here, $\Phi_\pi(Z_\tau^c,Z_\tau)$ is simply the probability a benchmark sample path $g$ incurs a cost greater than $\tau$, while the corresponding distancing sample path $h$ costs less than $\tau$, where $(h,g) \sim \Phi_\pi$. 
The difference \eqref{eq:thm_gap} encodes many complex dependencies on the epidemic parameters $\delta$ and $\beta$, the awareness weight $\alpha$, and the graphs $\mcal{G}_C$ and $\mcal{G}_I$.
The following result establishes stochastic domination of the distancing chain by the benchmark chain. 
\begin{corollary}
	The benchmark chain stochastically dominates the distancing chain on sample paths, i.e. $\mu_\pi \succeq \nu_\pi$. 
\end{corollary}
\begin{proof}
	For any upper set $\mcal{I} \subset \Gamma$, $\chi_{\mcal{I}}(\cdot)$ is increasing in $\Gamma$. By \eqref{eq:thm_gap}, 
	\begin{align}
		\mu_\pi(\mcal{I}) - \nu_\pi(\mcal{I}) &= \mbb{E}_{\mu_\pi}(\chi_{\mcal{I}}) - \mbb{E}_{\nu_\pi}(\chi_{\mcal{I}}) \nonumber \\
		&= \Phi_\pi(\mcal{I}^c,\mcal{I}) \geq 0 \label{eq:upper_difference} 
	\end{align}
\end{proof}

The difference in probability \eqref{eq:upper_difference} confirms the intuition gained from Corollary \ref{MFA_corollary} that sample paths with consistently high numbers of infected individuals are more probable under the benchmark chain. The closed-form differences \eqref{eq:thm_gap} and \eqref{eq:upper_difference} provide a stochastic analogue to Corollary \ref{MFA_corollary}, which only establishes inequality between the mean-field metastable states of benchmark and distancing models. 

Some examples of increasing $\mbb{Z}_+$-valued random variables in $\Gamma$ are
\begin{itemize}
	\item The absorption time $T:\Gamma \rightarrow \mbb{Z}_+$, defined by \eqref{eq:T}.
	\item The social cost up to time $m$, defined by $g \mapsto \sum_{t=0}^m |g^t|$, where $|s| \triangleq \sum_{i\in\mcal{N}} s_i$ for $s\in\Omega$.
	\item The ``epidemic spread", or how many unique nodes that contract the disease in a given amount of time $m$. This is given by $g \mapsto \sum_{i\in\mcal{N}}\chi_{E_i}(g)$, where $E_i = \{g :  \sum_{t=0}^m g_i^t > 0\}$. This metric is investigated on different network structures in the next section.
\end{itemize}
\begin{figure*}[t]
	\centering
	\begin{subfigure}{.32\textwidth}
		\centering
		\includegraphics[scale=.32]{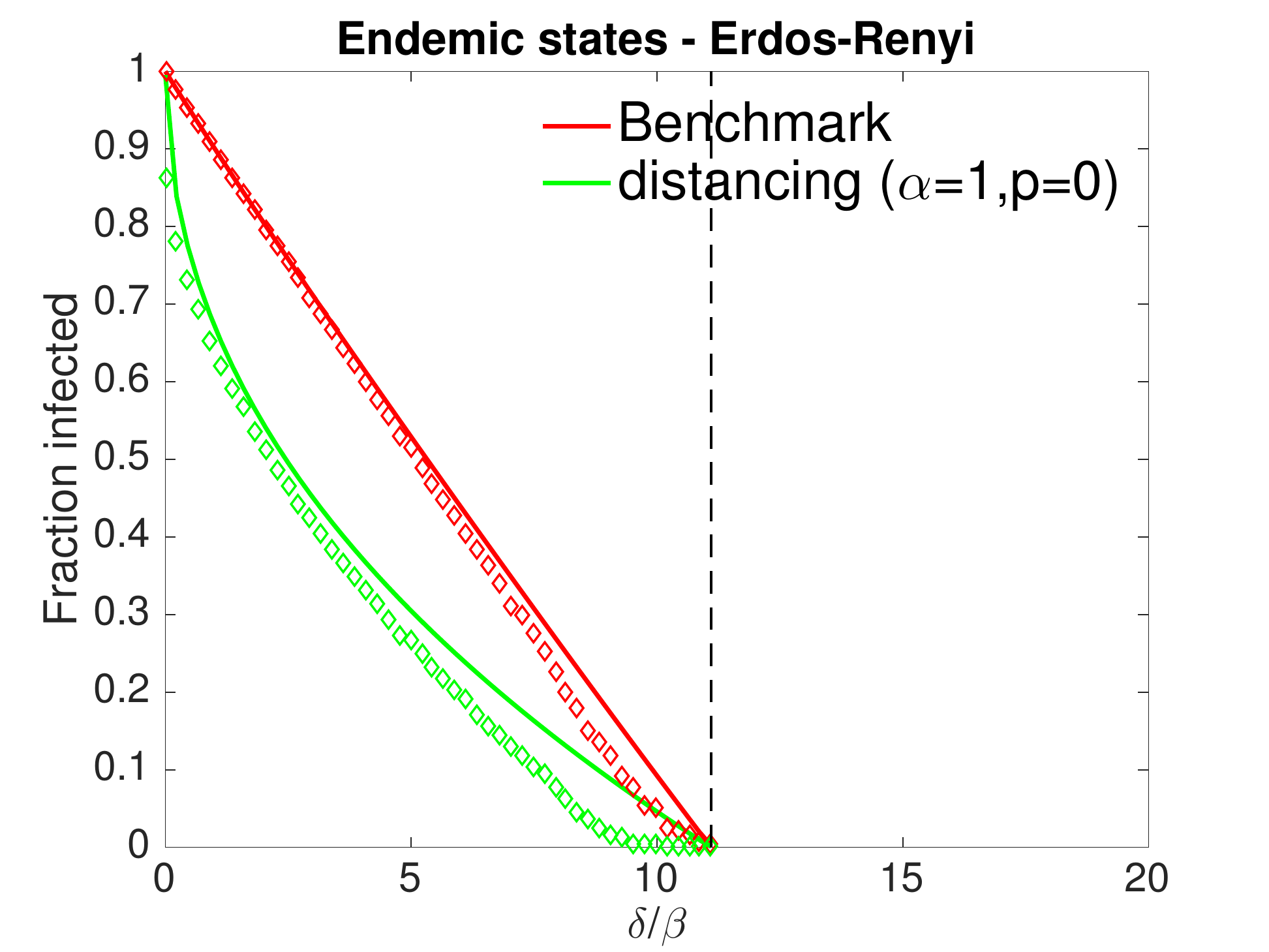}
		\caption{}
		\label{fig:FP_ER}
	\end{subfigure}  
	\begin{subfigure}{.32\textwidth}
		\centering
		\includegraphics[scale=.32]{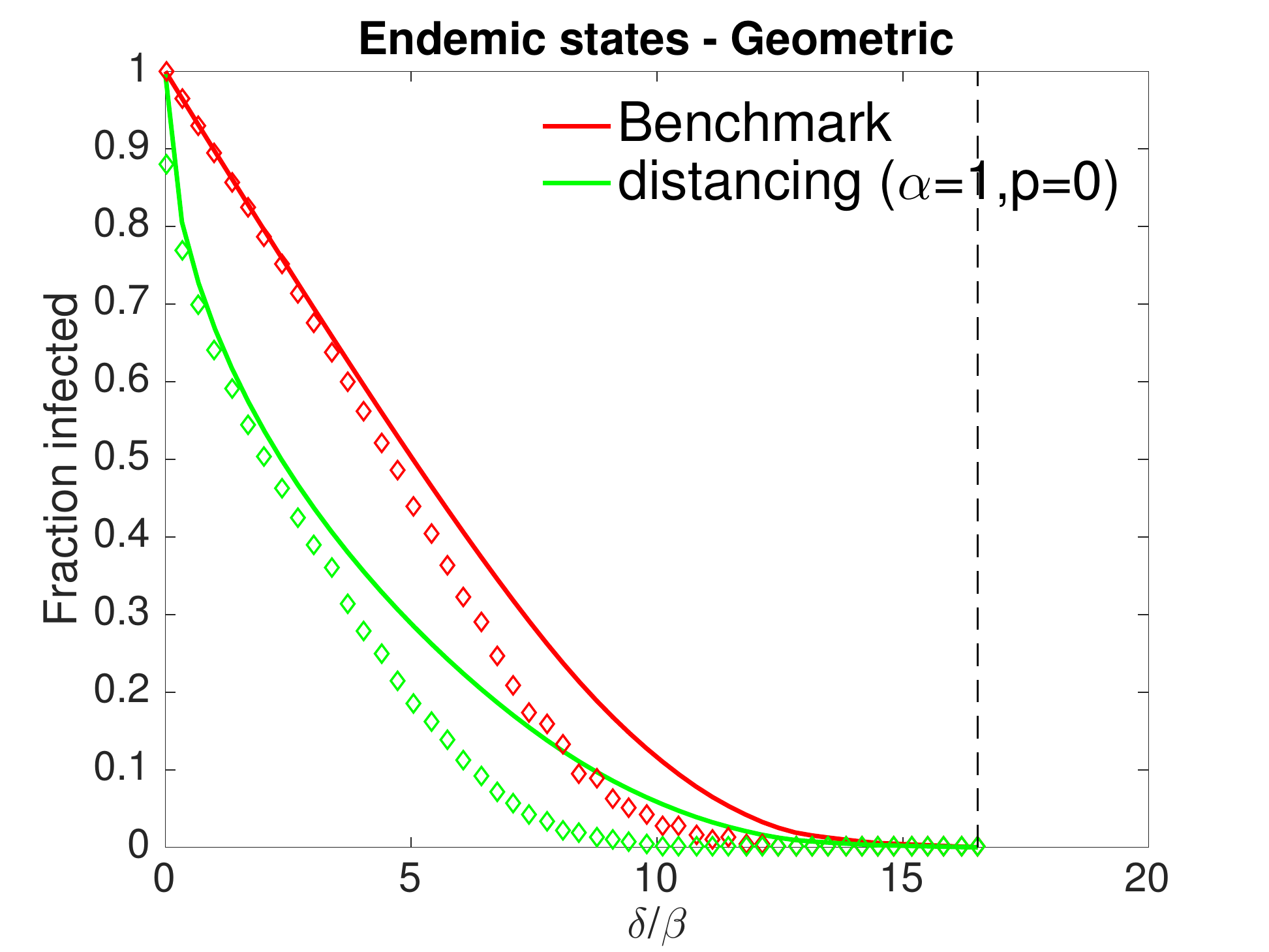}
		\caption{}
		\label{fig:FP_geo}
	\end{subfigure}
	\begin{subfigure}{.32\textwidth}
		\centering
		\includegraphics[scale=.32]{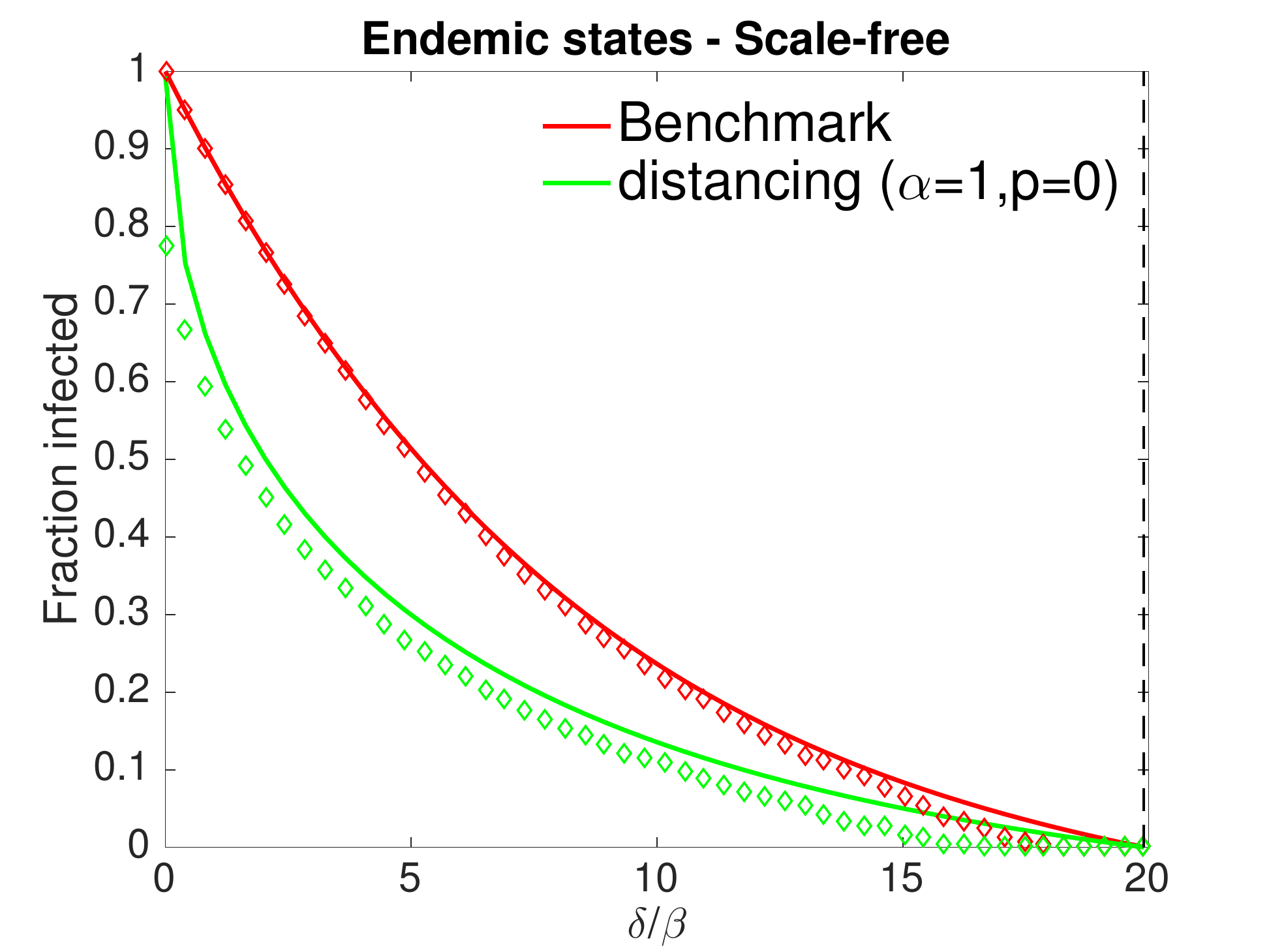}
		\caption{}
		\label{fig:FP_SF}
	\end{subfigure}
	\caption{\small Norms of the nontrivial fixed points (solid lines) and long-run fraction of infected in stochastic simulations (diamonds) in the range of epidemic persistence, $\delta/\beta \in [0,\lambda_{\text{max}}(A_C)]$, for $n=1000$ node networks. The fixed points are computed by iterating the MFA dynamics \eqref{eq:MFA_dist} and \eqref{eq:MFA_standard} with an arbitrary initial condition until convergence. The stochastic long-run infected fractions are computed by averaging the levels of epidemic states in the latter half of a sample run of length 200. Vertical dashed lines indicate $\lambda_{\text{max}}(A_C)$. (a) Erd\H{o}s-Renyi random network with $p_{\text{ER}}=.01$, $\lambda_{\text{max}}(A_C)=11.1$. Here, $p_{\text{ER}} > \log n/n$, the regime where the network is connected with high probability. (b) Geometric random graph with $r = .0564$, $\lambda_{\text{max}}(A_C)=16.52$. (c) Scale-free generated from the PA algorithm with $m=5$, $\lambda_{\text{max}}(A_C)=19.9$. The parameters are chosen such that all networks have the same average degree $d \approx 10$.} 
	\label{fig:FP_plots}
\end{figure*}
\begin{figure*}[t]
	\centering
	\begin{subfigure}{.32\textwidth}
		\centering
		\includegraphics[scale=.32]{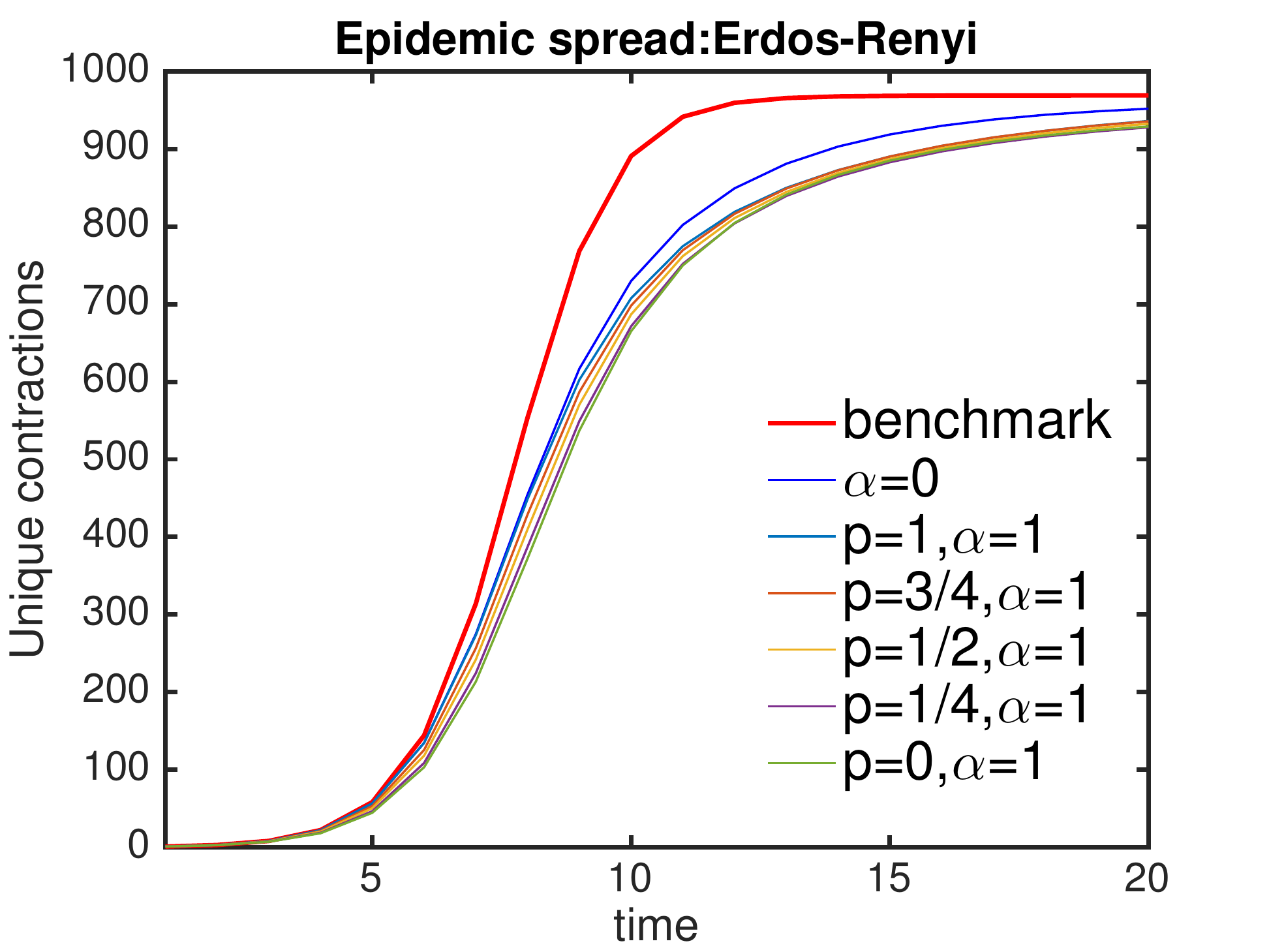}
		\caption{}
		\label{fig:ER_TS}
	\end{subfigure}  
	\begin{subfigure}{.32\textwidth}
		\centering
		\includegraphics[scale=.32]{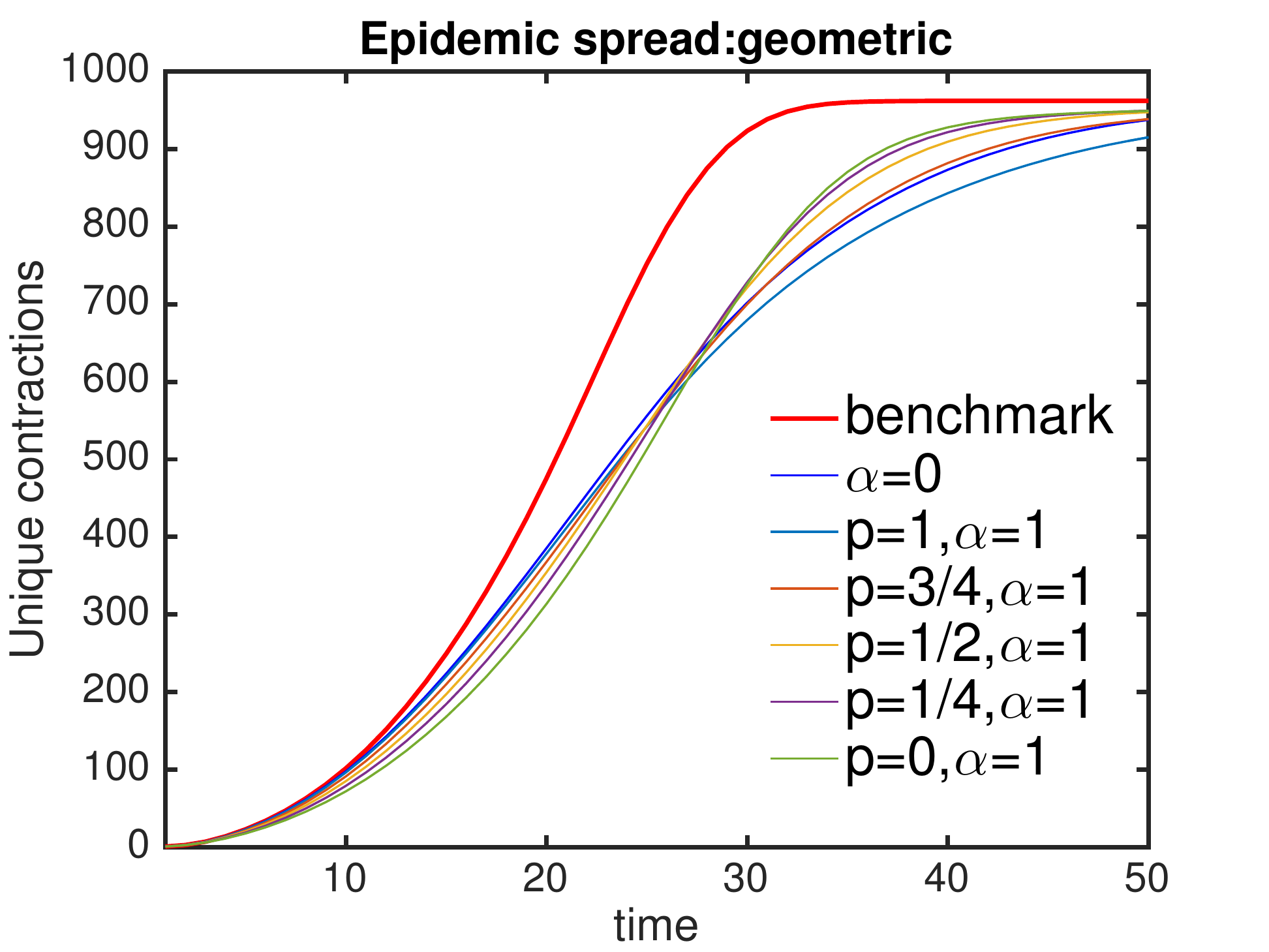}
		\caption{}
		\label{fig:geo_TS}
	\end{subfigure}
	\begin{subfigure}{.32\textwidth}
		\centering
		\includegraphics[scale=.32]{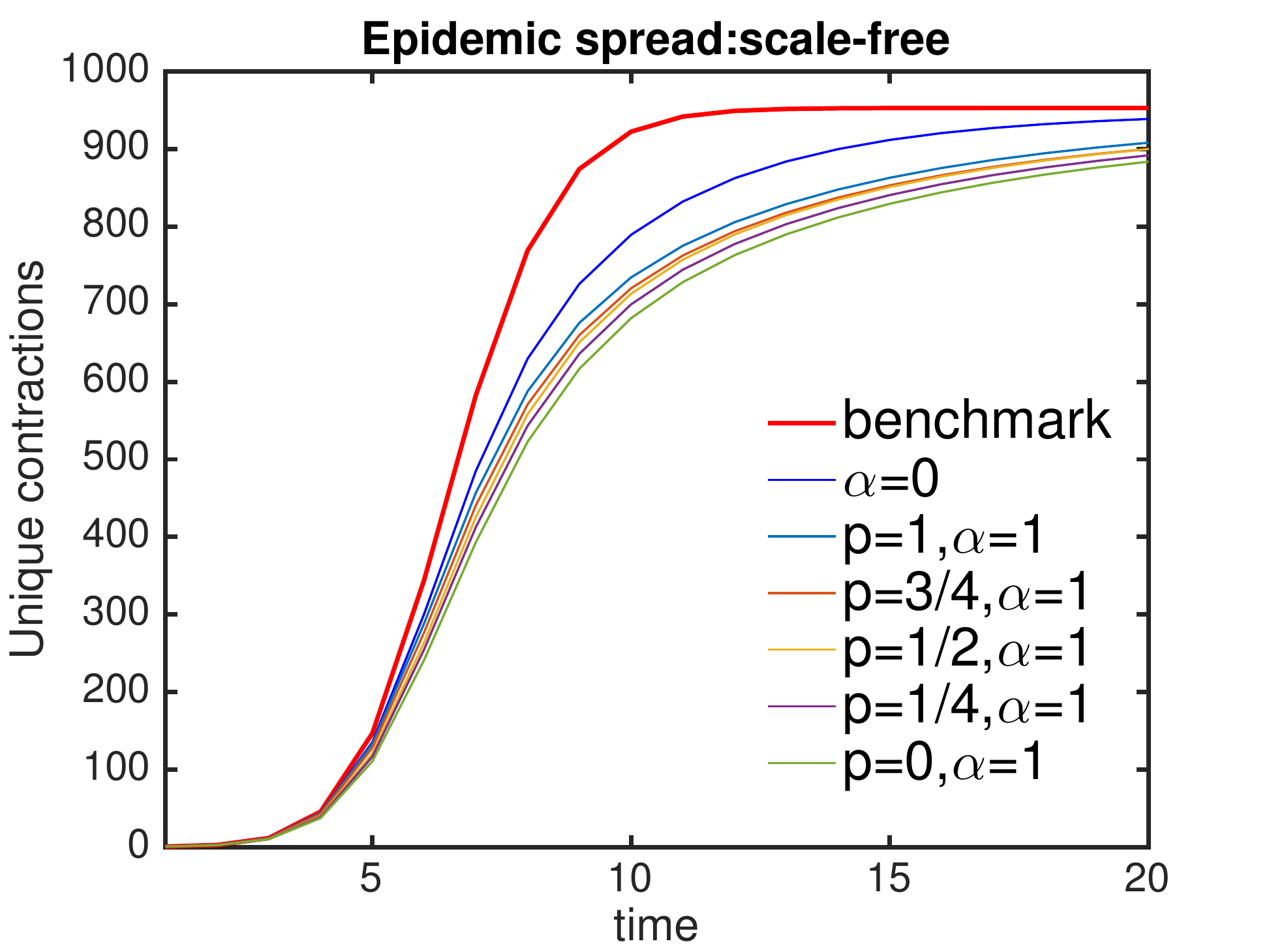}
		\caption{}
		\label{fig:SF_TS}
	\end{subfigure}
	\caption{\small Epidemic spreading as a function of time (same networks as Fig. \ref{fig:FP_plots}). Local contact information ($\alpha$ near 1, $p$ near 0) slows spread most effectively for (a),(c), and the early stages of (b), whereas global information ($\alpha$ near 0 or $\alpha=1,p$ near 1) is least effective. Note the inversion of awareness effectiveness in (b). In these simulations, $\delta = \beta = 0.2$. }
	\label{fig:spread_plots}
\end{figure*}
%
\section{Simulations on random networks}\label{sec:sims}

In this section, we illustrate through numerical simulations how the structure of the contact network $\mcal{G}_C$ affects the course of an epidemic in the presence of awareness and social distancing. Extensive analytical and simulation studies have been conducted without awareness (\cite{Chak2003},\cite{Ganesh2005},\cite{scalefree_threshold},\cite{VM2009},\cite{Vespignani_2001}). Here, we look at three random graph families - geometric, Erd\H{o}s-Renyi, and scale-free. These networks are relevant in studying epidemic spreading because they exhibit a variety of qualitative features that reflect real-world networks. Geometric networks portray people connected by geographic distance. Erd\H{o}s-Renyi random networks display a small-world effect common in many real world networks - e.g neural and social influence networks. Online social networks and the World Wide Web are examples of scale-free networks \cite{Newman}. 

In our model, the social network $\mcal{G}_I$ is generated directly from $\mcal{G}_C$ via a parameter $p\in(0,1)$ through the following procedure: 1) Select a fraction $p$ of existing edges in $\mcal{E}_C$ at random and remove them from the edge set; 2) For each of the selected edges, select one of the two end nodes randomly (e.g with probability 1/2) as the root node; 3) For each of the selected root nodes $i$, select $j\neq i$ uniformly at random and add the edge $(i,j)$. For $p$ close to one, the resulting graph $\mcal{G}_I = (\mcal{N},\mcal{E}_I)$ exhibits the small-world effect (small average shortest path length and small clustering) \cite{watts1998}. When $p=0$, $\mcal{G}_I = \mcal{G}_C$.

For the contact networks, geometric random graphs are generated by placing $n$ points uniformly at random on the unit torus (unit square with periodic boundary conditions). An edge exists between any two points if they are less than a specified distance $r\in(0,1)$ away. Erd\H{o}s-Renyi random graphs are constructed by forming an edge between any two nodes independently with a fixed probability $p_{\text{ER}}\in(0,1)$. Scale-free networks are generated by the preferential attachment algorithm \cite{Barabasi_1999}: starting with an initial connected graph of $m_0 \geq m$ nodes, $n-m_0$ additional nodes are added sequentially with each incoming node establishing links to $m$ existing nodes in the network. The probability a node receives an incoming link is proportional to its degree. We performed simulation analysis on one network from each random graph family. The networks all have 1000 nodes with an average degree of 10, and hence the same number of edges.

In Figure \ref{fig:FP_plots}, the normalized non-trivial fixed points are characterized for the three random networks in the interval of epidemic persistence. The norms of these points indicate the size of the endemic states and they slightly overestimate the actual long-run infected fraction observed in stochastic simulations of the Markov chains.

In Figure \ref{fig:spread_plots}, we quantify the ``epidemic spread" by the number of unique nodes that contract an infection as time progresses when one uniformly random node is initially infected. This metric is an example of an increasing random variable over sample paths (Section \ref{sec:coupling}) and is helpful in revealing not only how fast an epidemic initially spreads in the network, but also how far-reaching it is. A key observation is that contact awareness ($p=0,\alpha=1$) slows epidemic spreading better than any other awareness configuration at the beginning of an epidemic. This is intuitively clear since contact awareness provides nodes with the most vital information if they are in danger of getting infected. As $p$ increases, $\mcal{G}_I$ deviates more from $\mcal{G}_C$ and the information nodes receive become less vital.

Erd\H{o}s-Renyi and scale-free (with $m=5$) networks admit disease spread throughout the entire network in a short amount of time, even with social distancing (Figure \ref{fig:ER_TS},\ref{fig:SF_TS}). This is attributed to small average shortest path lengths (Ch. 8 \& 12, \cite{Newman}), allowing the epidemic to quickly spread to other parts of the network. Random geometric networks are characterized by high clustering and large diameter. Clustering slows the spread of an epidemic (Figure \ref{fig:geo_TS}), but also contributes to increasing the final epidemic size \cite{Volz_2011}. The virus stays localized and spreads slowly. This explains the inversion of awareness parameters in Figure \ref{fig:geo_TS}. By the time the epidemic first reaches its endemic level around $t=20$, many nodes have not yet been exposed because at this point their local communities are untouched. Thus, having global or long-range social awareness (low $\alpha$ or high $p$) is more beneficial over contact awareness.


\section{Conclusions}\label{sec:conclusion}

We modified the benchmark networked SIS epidemic process to include agent awareness, where prevalence-based information comes from social contacts and a global broadcast of the overall infected fraction. Agents take social distancing actions based on the level of information received, which reduces their probabilities of getting infected. We showed that awareness does not change the epidemic threshold for persistence by proving existence of a nontrivial fixed point in the mean-field approximation. Any nontrivial fixed point of the distancing model is strictly component-wise less than the unique nontrivial fixed point of the benchmark model.

We provided a full stochastic comparison analysis between the benchmark and distancing chains in terms of their respective probability distributions on sample paths by constructing a monotone coupling. The construction relies on exploiting the differences in node transition probabilities between the two chains. Consequently, adding awareness reduces the expectation of any increasing random variable on sample paths and we obtain a closed form expression for the reduction. This implies the benchmark distribution on sample paths stochastically dominates the distancing distribution.

In simulations, we showed epidemic spreading heavily depends on the network structure. In particular, qualitative features such as small-world effects, clustering, and diameter explain the results seen in simulations. We also concluded local contact awareness is the most effective at slowing epidemic spread, and global awareness is the least effective.
 
\appendix

\begin{proof}[Proof of Theorem \ref{Phi_coupling}]

When $x\preceq_\Omega y$, the $\varphi^{x,y}_i$ are well-defined probabilities since $p_{01}^i(y) - p_{01,d}^i(x) \geq 0$ and $1 - p_{01,d}^i(x) - \delta(1 - p_{01}^i(y)) > 0$. One can see by inspection that $\varphi^{x,y}_i$ is a monotone coupling of $\mbb{P}_i^d(x,\cdot)$ and $\mbb{P}_i(y,\cdot)$ defined in \eqref{eq:distancing_measure}, \eqref{eq:distancing_measure2} and \eqref{eq:bench_measure}, \eqref{eq:bench_measure2} respectively.

Observe from \eqref{eq:phi_omega}, $\varphi^{x,y}(\omega,\omega') > 0$ implies $x\preceq_\Omega y$ and $\omega \preceq_\Omega \omega'$. Consequently, $\varphi^{x,y}$ is a monotone coupling of $K_d(x,\cdot),K(y,\cdot)$:
\begin{align}
	\sum_{\omega'\succeq_\Omega \omega} \varphi^{x,y}(\omega,\omega') &= \prod_{i=1}^n \left(\sum_{\omega'_i\geq\omega_i} \varphi^{x,y}_i(\omega_i,\omega'_i) \right) \\
	&= \prod_{i=1}^n \mbb{P}_i^d(x,\omega_i) \\
	&= K_d(x,\omega)
\end{align}
By a completely analogous computation, we obtain $\sum_{\omega'\preceq_\Omega\omega} \varphi^{x,y}(\omega',\omega) = K(y,\omega)$. 

Also, notice from \eqref{eq:phi_pi} that $\Phi_\pi(h,g) > 0$ implies $h^0 = g^0$ and $h \preceq_\Gamma g$. Consequently, $\Phi_\pi$ is a monotone coupling of $\nu_\pi,\mu_\pi$:
\begin{align}
	\sum_{g\succeq_\Gamma h} \Phi_\pi(h,g) &= \sum_{\substack{g\succeq_\Gamma h \\ g^0 = h^0 }} \pi(h^0)\prod_{t=0}^{T(g)-1} \varphi^{h^t,g^t}(h^{t+1},g^{t+1}) \label{from_def}  \\
	&= \pi(h^0)\prod_{t=1}^{T(h)} \sum_{g^t \succeq_\Omega h^t} \varphi^{h^{t-1},g^{t-1}}(h^t,g^t) \label{eq:combinatorial_gamma} \\
	&= \pi(h^0)\prod_{t=1}^{T(h)} K_d(h^{t-1},h^t) \\
	&= \nu_\pi(h)
\end{align}
The equality \eqref{eq:combinatorial_gamma} is the combinatorial form of writing \eqref{from_def}, and the product terminates at $T(h)$ because 1) $g\succeq_\Gamma h$ implies $T(g) \geq T(h)$, 2) $h^t = \bold{o}$ for all $t \geq T(h)$ and 3) for any $t > T(h)$,
\begin{align*}
	\sum_{g^t \succeq_\Omega h^t} \varphi^{h^{t-1},g^{t-1}}(h^t,g^t) &= \sum_{g^t\in\Omega} \varphi^{\bold{o},g^{t-1}}(\bold{o},g^t) = K_d(\bold{o},\bold{o}) \\
	&=1
\end{align*}
By an analogous computation, $\sum_{h\preceq_\Gamma g} \Phi_\pi(h,g) = \mu_\pi(g)$. 
\end{proof}
\section*{Acknowledgements}

This work is supported by ARO grant \#W911NF-14-1-0402, and supported in part by KAUST. The authors thank J. Walker Gussler (Georgia Inst. Tech.) for his contribution in the simulations.






\bibliographystyle{IEEEtran}
\bibliography{sources}

\end{document}